\documentclass[twocolumn,prl,aps,showpacs,superscriptaddress,tighten]{revtex4}
\usepackage{xspace,amsmath,amsfonts,amsthm,amssymb,amsbsy,graphicx,color}

\begin{document}

\title{Locally Optimal Control of Quantum Systems with Strong Feedback}
\author{Alireza Shabani}
\affiliation{Department of Electrical Engineering, University of
Southern California, Los Angeles, CA 90089, USA}
\author{Kurt Jacobs}
\affiliation{Department of Physics, University of Massachusetts at Boston, 100 Morrissey
Blvd, Boston, MA 02125, USA}

\begin{abstract}
For quantum systems with high purity, we find all observables that, when continuously monitored, maximize the instantaneous reduction in the average linear entropy. This allows us to obtain all locally optimal feedback protocols with strong feedback, and explicit expressions for the best such protocols for systems of size $N \leq 4$. We also show that for a qutrit the locally optimal protocol is the \emph{optimal} protocol for observables with equispaced eigenvalues, providing the first fully optimal feedback protocol for a 3-state system. 
\end{abstract}

\maketitle

\newtheorem{theo}{Theorem} \newtheorem{lemma}{Lemma}

Observation and control of coherent quantum behavior has been realized in a variety of mesoscopic devices~\cite{Steffen06, Majer07, Niskanen07, Plantenberg07}. With further refinements, such devices may well form the basis of new technologies, for example in sensing~\cite{Boixo07} and communication~\cite{Cook07, Chen07}. Feedback, in which a system is continuously observed and the information used to control its behavior in the presence of noise, is an important element in the quantum engineer's toolbox~\cite{Cook07, Chen07, Smith02, rapidP, Steck04, Yanagisawa06}. In view of this, one would like to know the limits on such control, given any relevant limitations on the measurement and/or control forces. However, except in special cases~\cite{Yanagisawa98}, the dynamics of continuously observed quantum systems is nonlinear. Further, results on the quantum-to-classical transition show that this nonlinear dynamics, described by stochastic master equations (SME's), is necessarily every bit as complex (and chaotic) as that of nonlinear classical systems~\cite{Bhattacharya00}. Because of this, fully general and exact results regarding optimal quantum feedback are unlikely to exist; certainly no such results have been found for nonlinear classical systems~\cite{Whittle}. Nevertheless, one would like to obtain results that give insights applicable across a range of systems.

Quantum feedback control is implemented by modifying a ``control" Hamiltonian, $H$,  that is some part of the system Hamiltonian. Here we will examine feedback protocols in the regime where the controls are able to keep the system close to a pure state. This is an important regime, both because it is where many quantum control systems will need to operate, and because it allows one to simplify the problem by using a power series expansion~\cite{Jacobs07c}. In addition to working in the regime of good control, we make two further simplifications. The first is that the control is strong -- that is, that 1) the only constraint on $H$ is that $\mbox{Tr}[H^2] \leq \mu^2$ for some constant $\mu$, and 2) that $H$ can induce dynamics much faster than both the dynamics of the system and the rate at which the measurement extracts information. This means that $H$ is {\em effectively} unconstrained. We thus deal strictly with a subset of the regime of good control, defined by $\mu \gg k$ and $k \gg \beta$. Here $k$ is the strength of the measurement (defined precisely below), and $\beta$ is the noise strength, which we define as the rate of increase of the linear entropy due to the noise. The latter inequality is essential for good control.  This regime is applicable, for example, to mesoscopic superconducting systems~\cite{Steffen06, Majer07, Niskanen07, Plantenberg07}, such as coupled Cooper-pair boxes. Here the speed of control rotations is typically 1-10 GHz~\cite{Steffen06}, and that of decoherence  is $10^6 \mbox{s}^{-1}$~\cite{Vion02}. Measuring these at a rate $k = 5\times 10^7 \mbox{s}^{-1}$ is reasonable~\cite{Jacobs07b}, and falls in the above regime. 

Our second simplification is to seek control protocols that give the maximum increase in the control objective in each time-step \emph{separately} --- that is, that are \emph{locally optimal} in time. However, we will find that for a qutrit, the locally optimal protocol (LOP) is the {\em optimal} protocol for observables with equispaced eigenvalues. 

We will allow the controller to measure a single observable, $X$. Since the control allows us to perform all unitary operations, and since transforming the system is equivalent to transforming the observable being measured, $X$, this allows the controller to measure all observables of the form $X^{\mbox{\scriptsize u}} = UXU^\dagger$, for any unitary $U$. Since the control Hamiltonian is not limited, the only constraint on the controller
is the rate at which the measurement extracts information (the measurement strength, $k$).

A sensible and widely applicable control objective is to maximize the probability, $P$, that the system will be found in a desired pure state (referred to as the {\em target} state) at a given time $T$ (called the {\em horizon} time). This objective also allows one to maximize $P$ in the steady-state, and is the objective we will consider here. 

In what follows we will denote the state of the quantum system by the density matrix $\rho$, and the $N$ eigenvalues of $\rho$ as $\lambda_i$. We place these in decreasing order so that $\lambda_{i}\geq\lambda_{i+1}$. Since the control dynamics is fast, at any time $T$ we can apply $H$ to quickly rotate $\rho$ so as to maximize $P$ at that time. This means rotating $\rho$ so that the eigenstate corresponding to $\lambda_0$ is the target state, giving $P = \lambda_0$. Thus the optimality of the control is determined entirely by the eigenvalues $\lambda_i$.  (Since the control Hamiltonian cannot change these eigenvalues, the only further role of $H$ is to set the observable to be measured at each time, $X^{\mbox{\scriptsize u}}(t)$.) The probability that the state is found in the target state at a given future time, is thus the average of $\lambda_0$ over all future trajectories at that time: $P(t) = \langle \lambda_0(t) \rangle$. Note also that because the state is near-pure, we can write $\lambda_0 = 1 - \Delta$, with $\Delta \ll 1$. Thus $P = 1 - \langle \Delta(t) \rangle$, with $\langle \Delta (T) \rangle $ the error probability.

By definition, the locally optimal feedback protocol is the one that maximizes the rate of reduction of $\langle \Delta \rangle$ at each time-step. To find the LOP, we must
find the observable, $X^{\mbox{\scriptsize u}}$, that maximizes the rate of
reduction of $\langle \Delta \rangle$ for any $\rho$. To derive the equation of motion for $\langle\Delta\rangle$, we start with the SME for the density matrix under a continuous measurement of $X^{\mbox{\scriptsize u}}$: 
\begin{equation}
d\rho =-k[X^{\mbox{\scriptsize u}},[X^{\mbox{\scriptsize u}},\rho ]]dt+\sqrt{2k}(X^{\mbox{\scriptsize u}}\rho +\rho X^{\mbox{\scriptsize u}}-2\langle X^{\mbox{\scriptsize u}}\rangle\rho )dW .
\label{SME}
\end{equation}
Here the assumption of strong feedback allows us to drop any system
Hamiltonian, and $dW$ is Gaussian (Wiener) noise, satisfying $\langle
dW\rangle =0$ and $dW^{2}=dt$. Note that this SME does not include noise; we exclude noise in what follows except when we calculate results for the steady-state.  To obtain the equation of motion for $\langle \Delta \rangle$, to first order in $\Delta$, we first note that $\Delta = (1 - \mbox{Tr}[\rho^2])/2$ (that is, half the linear entropy), to first order in $\Delta$~\footnote{Because the linear entropy is equal to $2\Delta$ to first order in $\Delta$,  the protocol in~\cite{rapidP} is the optimal protocol for our problem for a single qubit. Our purpose here is, of course, to go beyond a single qubit. Our results are also relevant to rapid-purification protocols for larger systems~\cite{rap2}.}. We calculate the derivative of $\mbox{Tr}[\rho^2]$ directly from the SME, and then expand the result in powers of $\Delta$. This gives  
\begin{eqnarray}
  \!\!\!\!\!\!\!\! d\Delta & \!\! = \!\! & - 8k\sum_{i\not=0}\lambda _{i}|X^{\mbox{\scriptsize u}}_{i0}|^{2} dt  \! - \! \sqrt{8k} (\Delta X^{\mbox{\scriptsize u}}_{00} \! - \!\! \sum_{i\not=0}\lambda _{i} X^{\mbox{\scriptsize u}}_{ii} ) dW ,
\end{eqnarray}
where $X^{\mbox{\scriptsize u}}_{nm} \equiv \langle
n|X^{\mbox{\scriptsize u}}|m\rangle$. The equation of motion for $\langle \Delta \rangle$ is given by averaging this equation over $dW$. Thus  $\langle \dot{\Delta}\rangle =-8k\sum_{i\not=0}\lambda _{i}|X^{\mbox{\scriptsize u}}_{i0}|^{2}$. 

We now prove the following theorem:

\begin{theo}
Define $X$ as Hermitian, $U$ as unitary, and $\rho$ as a density operator of dimension $N$. Without loss of generality we set $X$ and $\rho$ to be diagonal, arrange the eigenvalues of $\rho$, $\lambda_i$, in decreasing order, and arrange the eigenvalues of $X$ so that the two extreme values are in the first 2-by-2 block, corresponding to the two largest eigenvalues of $\rho$~\footnote{If the extreme eigenvalues of $X$ are degenerate, then one has a choice as to which eigenvectors to select to correspond to the extreme eigenvalues.}. The maximum of $F(U) = \sum_{i\not=0} \lambda_i  | \langle i|U X U^\dagger | 0\rangle |^2 \equiv  \sum_{i\not=0} \lambda_i  |X^{\mbox{\scriptsize u}}_{i0} |^2$ is achieved if
\begin{equation}
   U =  U_{\mbox{\scriptsize {\em opt}}} =  U_2^{\mbox{\scriptsize {\em u}}} \oplus V ,
   \label{Uopt}
\end{equation}
where $U_2^{\mbox{\scriptsize {\em u}}}$ is any 2-by-2 unitary unbiased w.r.t the basis $\{(1,0),(0,1)\}$:
\begin{equation}
U_2^{\mbox{\scriptsize {\em u}}} =  \frac{e^{i \phi}}{\sqrt{2}} \left(
\begin{array}{cc}
 e^{i\theta_1} &  -e^{i\theta_2}    \\
 e^{-i\theta_2} &   e^{-i\theta_1}
\end{array} \right) ,
\end{equation}
and $V$ is any unitary with dimension $N \! - \! 2$. 
\label{thm1}
\end{theo}

\begin{proof}
We first derive an upper bound on $F(U)$. This is 
\begin{eqnarray} 
F(U) & \le &  \lambda_1 \sum_{i=1}^{N-1} |X^{\mbox{\scriptsize u}}_{i0} |^2 = \lambda_1 \left[ \sum_{i=0}^{N-1} |X^{\mbox{\scriptsize u}}_{i0}|^2 - |X^{\mbox{\scriptsize u}}_{00} |^2  \right]  \nonumber \\ 
& = &  \lambda_1 \mbox{Var}\left( X, U^\dagger | 0> \right) \le  \lambda_1 \frac{(x_{max} - x_{min})^2}{4}.  
\end{eqnarray}
Here $\mbox{Var}(X,|\psi\rangle)$ denotes the variance of $X$ in the state $|\psi\rangle$. The first inequality is immediate, and the last is well-known~\cite{Seo}. Since $U_{\mbox{\scriptsize opt}}$ saturates this bound, it achieves the maximum. That only unitaries of the above form achieve the maximum is simplest to show when the eigenvalues of $\rho$ are non-degenerate: to saturate the first inequality one must restrict $U$ to the subspace spanned by $\{ |0\rangle, |1\rangle\}$, and to achieve the last, $U$ must be unbiased w.r.t to the eigenbasis. When the eigenvalues of $\rho$ are degenerate, a careful analysis shows that this remains true~\footnote{See EPAPS Document No. E-PRLTAO-101-012850 for a constructive proof of theorem 1. For more information on EPAPS, see http://www.aip.org/pubservs/epaps.html.}.      
\end{proof}

The remarkably simple form of $U_{\mbox{\scriptsize opt}}$ tells us that to 
obtain the fastest reduction in $\langle \Delta\rangle$ at each time, $t$, and thus 
realize an LOP, we must choose $X^{\mbox{\scriptsize u}}$ at each time 
to concentrate the distinguishing power of the measurement entirely on the largest two eigenvalues of $\rho$ at that time, and measure in a basis that is \emph{unbiased} with respect to the corresponding eigenvectors. It also tells us that the maximum achievable rate of reduction is 
\begin{equation}
\langle\dot{\Delta}\rangle = -8k \langle \lambda_1\rangle |X^{
\mbox{\scriptsize u}}_{01}|^2 = -2k\langle \lambda_1\rangle (x_{
\mbox{\scriptsize max}} - x_{\mbox{\scriptsize min}})^2 ,
\end{equation}
where $\lambda_1$ is the
second largest eigenvalue of $\rho(t)$, and $x_{\mbox{\scriptsize max}}$ and $
x_{\mbox{\scriptsize min}}$ are the maximum and minimum eigenvalues of $X$.
Note that $\langle \lambda_1\rangle$ also decreases at the rate $\langle\dot{\Delta}\rangle$. We have complete freedom in choosing the unitary submatrix $V$, as
it has no effect on $\langle\dot{\Delta}\rangle$. However, $V$ also 
has no effect on $\lambda_1$; $V$ induces transition rates only between the 
$(N-2)$ smallest eigenvalues. 

We can now obtain a lower bound on the performance of LOP's for any system.
Whatever the choice of $V$, $\langle \lambda_1\rangle$ will always be
greater than or equal to $\langle \Delta\rangle/(N-1)$. We therefore have 
$\langle\dot{\Delta}\rangle \leq -[8/(N-1)] k \langle \Delta\rangle |X^{
\mbox{\scriptsize u}}_{01}|^2$, and thus in the absence of noise, throughout
the evolution the error probability will satisfy
\begin{equation}
\langle \Delta(t)\rangle \leq \Delta_0 e^{-2kt (\delta x)^2/(N-1)} . \label{lb1}
\end{equation}
where $\delta x \equiv x_{\mbox{\scriptsize max}} - x_{ 
\mbox{\scriptsize min}}$.

In the presence of noise, the important quantity is the steady-state error probability, $\langle \Delta \rangle_{\mbox{\scriptsize ss}}$, and we can re-derive the lower bound on this given in~\cite{Jacobs07c}. In the worst case, $V$ leaves the $N-2$ smallest eigenvalues unchanged, so that under
isotropic noise all the small eigenvalues remain identical once
homogenized by the action of the LOP. The equation of motion for $\langle \Delta  \rangle_{\mbox{\scriptsize ss}}$ is then $\langle \dot{\Delta}\rangle =-8k\langle \Delta \rangle
|X_{01}^{\mbox{\scriptsize u}}|^{2}/(N-1)+\beta /2$ (recall that $\beta $ is the noise strength). This gives $ \langle \Delta \rangle _{\mbox{\scriptsize ss}} = [\beta (N-1)]/[4k(\delta x )^{2}]$, a lower bound on the performance of all LOP's with isotropic noise.

We further have the nice result that, for qubits and qutrits, the two lower bounds just derived are tight -- here they give the performance of the best LOP's because the action of $V$ is trivial for $N<4$. 

For $N\geq 4$, to obtain the \emph{best} LOP, one would need to choose $V$
to continually minimize the entropy of the smallest $N-2$ eigenvalues in
such a way as to allow one to generate the largest possible value of 
$|\langle\dot{\Delta}\rangle|$ at all future times. Using time-independent
perturbation theory~\cite{BJBook}, we can derive from Eq.(\ref{SME}) the
equations of motion for all the small eigenvalues. These are
\begin{equation}
d\lambda_i = F_i(\boldsymbol{\lambda}, X^{\mbox{\scriptsize u}}) dt +
\sigma_i(\boldsymbol{\lambda},X^{\mbox{\scriptsize u}}) dW  , \label{eqsys}
\end{equation}
where $\boldsymbol{\lambda} = (\lambda_1,\ldots,\lambda_{N-1})$, and
\begin{eqnarray}
F_i & = & - 8 k \left[ \lambda_i |X^{\mbox{\scriptsize u}}_{0i}|^2 +
\sum_{j>0, j\not= i}\frac{\lambda_i \lambda_j}{\lambda_i - \lambda_j}
|X^{\mbox{\scriptsize u}}_{ji}|^2 \right] ,  \notag \\
\sigma_i & = & \sqrt{8k} \lambda_i ( X^{\mbox{\scriptsize u}}_{00} - X^{
\mbox{\scriptsize u}}_{ii}) ,
\end{eqnarray}
for $i = 1,\ldots , N-1$. Note that under locally optimal control, the
equation for $\langle \lambda_1\rangle$ reduces to $d\langle
\lambda_1\rangle = -8k\langle \lambda_1\rangle |X^{\mbox{\scriptsize u}
}_{01}|^2$.

The equations for the small eigenvalues are nonlinear. As a result of this,
in general in finding the optimal $V$, one cannot easily eliminate the
stochastic terms as we have been able to in the analysis so far, even though
we are interested purely in the average value of $\Delta$. Nevertheless for
$N=4$ we can obtain the optimal $V$ by combining the above results with
those of~\cite{rapidP}, which shows that the maximal increase in the largest
eigenvalue of a qubit is obtained when the observable is unbiased with
respect to the eigenvectors. Since the best thing we can do, given that we
continually maximize $dS$, is to separate the two smallest eigenvalues as
rapidly as possible, the result in~\cite{rapidP} tells us that $V$ must be
unbiased with respect to the eigenvectors of the two smallest eigenvalues.
We now label the eigenvalues of $X$ in decreasing order as $x_{i}$. Because the SME is invariant under the transformation $X \rightarrow X + \alpha I$ ($\alpha$ real), we add a constant to $X$ so that $x_{4}=-x_{1}$, without loss of generality. The best locally optimal control is then achieved by
\begin{equation}
X^{\mbox{\scriptsize u}}=\left(
\begin{array}{cc}
0 & x_{1} \\
x_{1} & 0
\end{array}
\right) \oplus \left(
\begin{array}{cc}
d & c \\
c & d
\end{array}
\right) ,
\end{equation}
where $2d=x_{2}+x_{3}$, $2c=x_{2}-x_{3}$, and $x_{1}>c$. 
If the eigenvalues of $X^{\mbox{\scriptsize u}}$ are equally
spaced, then $d=0$ and all stochastic terms vanish. In this case $\langle \Delta \rangle = \Delta$, and the equations for the system, excluding noise, are
\begin{eqnarray}
\dot{\lambda}_{1} &=&\dot{\Delta}\;=\;-8kx_{1}^{2}\lambda _{1} , \\
\dot{\lambda}_{2} &=&-\dot{\lambda}_{3}\;=\;8kc^{2}\lambda _{2}\lambda
_{3}/(\lambda _{2}-\lambda _{3}) ,
\end{eqnarray}
where $\dot{\Delta}=\sum_{i=1}^{3}\dot{\lambda}_{i}=\dot{\lambda}_{1}$. Even
though these equations are nonlinear, it is possible to obtain an analytic
expression for the behavior of $\Delta $ once certain transients have died
away. To do this note that the LOP first equalizes $\lambda_1$ and $\lambda_2$, and then must rapidly and repeatedly swap them. As a result they remain equal, and their derivatives become the average of $\dot{\lambda_1}$ and $\dot{\lambda_2}$ above. Next, calculating the derivative of the ratio $R=\lambda_{1}/\lambda _{3}$, we find that for $x_1 > \sqrt{2}c$, $R$ stabilizes at the value $R_{\mbox{\scriptsize ss}}=(x_{1}^{2}+c^{2})/(x_{1}^{2}-c^{2})$. Once this has happened, the equation for $\Delta $ reduces to the simple exponential decay $\dot{\Delta}=-\gamma \Delta $, with the rate
\begin{equation}
  \gamma =[4k/(3c^{2})](x_{1}^{2}-c^{2})(x_{1}^{2}-2c^{2}) \; , \;\;\;\;\; x_1 > \sqrt{2} c .
\end{equation}
For $c < x_1 \leq \sqrt{2} c$, after a time such that $R \gg 1$, the result is also exponential decay, but with $\gamma = 4 k x_1^2$. 

We have now found the best locally optimal protocols for $N=3$ and $4$, but
in each case the LOP is not necessarily the optimal protocol. We will now
examine the LOP for a qutrit and show that under certain conditions it is the
optimal protocol. Before we do this, we note that we can use the theorem above to
place an upper bound on the performance of \emph{any} protocol for all
systems in the regime of good control. Since $\mbox{max}(\langle \dot{\Delta}
\rangle )=-2k\langle \lambda _{1}\rangle (x_{\mbox{\scriptsize max}}-x_{
\mbox{\scriptsize min}})^{2}$, and $\langle \lambda_1 \rangle \leq \langle
\Delta \rangle $, the steady-state error probability for any protocol satisfies
\begin{equation}
\langle \Delta \rangle _{
\mbox{\scriptsize ss}}\geq \beta/[4k (x_{\mbox{\scriptsize max}}-x_{
\mbox{\scriptsize min}})^{2}] ,
\end{equation}
where $\beta $ is once again the noise strength. This is true for any noise process, isotropic or otherwise.

We now analyze the case of a qutrit when $X$ has equally spaced eigenvalues. As usual we denote these as $x_1 > x_2 >x_3$. We also add a constant to $X$ so that $x_3 = -x_1$ and $x_2=0$. With these definitions, the LOP for a single qutrit involves choosing $U$ so that 
\begin{equation}
X^{\mbox{\scriptsize u}}_1 = \left(
\begin{array}{ccc}
0 & q & 0 \\
q & 0 & 0 \\
0 & 0 & 0
\end{array}
\right) , 
\end{equation}
where $q = (x_1 - x_3)/2 = x_1$. This generates the evolution $
(\lambda_1(t),\lambda_2) = (\lambda_1^{0} e^{-\gamma t}, \lambda_2^{0})$,
where $\lambda_1^{0}$ and $\lambda_2^{0}$ are the initial eigenvalues of 
$\lambda_1$ and $\lambda_2$, and we have defined $\gamma \equiv 8k q^2$. 
This measurement is applied until $\lambda_1(t) = \lambda_2^{0}$, which 
occurs after a time $\tau = \ln{
(\lambda_1^{0}/\lambda_2^{0})}/\gamma$. At this point the LOP changes
abruptly, and involves rapidly switching the measurement between $X^{
\mbox{\scriptsize u}}_1$ and $X^{\mbox{\scriptsize u}}_2 = O_{
\mbox{\scriptsize flip}} X^{\mbox{\scriptsize u}}_1 O^\dagger_{
\mbox{\scriptsize flip}}$, where $O_{\mbox{\scriptsize flip}}$ swaps the
eigenstates of $\lambda_1$ and $\lambda_2$. In the limit of fast switching,
this generates the evolution $(\lambda_1(t),\lambda_2(t)) =
(\lambda_1(\tau),\lambda_2(\tau))e^{-\gamma t/2}$. Denoting now the \emph{initial} time by $t$, the error probability under the LOP at the final time $T$ (the
horizon time) is thus 
\begin{eqnarray}
\Delta_{\mbox{\scriptsize LOP}}(\boldsymbol{\Lambda},t,T) & = & \Lambda_1
e^{-\gamma (T-t)} + \Lambda_2 , \;\;\;\;\;\;\;\;\; T \leq t + \tau ,  \notag
\\
\Delta_{\mbox{\scriptsize LOP}}(\boldsymbol{\Lambda},t,T) & = &2 \sqrt{
\Lambda_1\Lambda_2} e^{-\gamma (T-t)/2} , \;\;\;\;\;\; \! T \geq t + \tau ,
\notag
\end{eqnarray}
where we have defined $\Lambda_1 \equiv \lambda_1(t)$, $\Lambda_2 \equiv
\lambda_2(t)$ and $\boldsymbol{\Lambda} \equiv (\Lambda_1,\Lambda_2)$. In optimal control theory, the quantity we wish to minimize, as a function of the initial and final times, is called the {\em cost} function. Having an explicit expression for the cost function generated
by the LOP, $\Delta_{\mbox{\scriptsize LOP}}(\boldsymbol{\Lambda},t,T)$, we can 
now use the verification theorems
of optimal control theory to determine whether the LOP is the optimal
protocol~\cite{Wiseman08,Zhou97}. We first consider the case when $T < t + \tau$.
For the LOP to be optimal, the cost function must satisfy the
Hamilton-Jacobi-Bellman (HJB) equation corresponding to the dynamical
equations for the system (Eq.(\ref{eqsys}))~\cite{Zhou97}:
\begin{eqnarray}
\frac{\partial \Delta}{\partial t} = \max_{X^{\mbox{\scriptsize u}}(t)} \left[ -
\sum_i F_i \frac{\partial \Delta}{\partial \Lambda_i} - \sum_{ij}\frac{\sigma_i
\sigma_j}{2} \frac{\partial^2 \Delta}{\partial \Lambda_i \partial \Lambda_j} 
\right] .
\end{eqnarray}
To check that the cost function is a solution to this equation, one
substitutes in $\Delta_{\mbox{\scriptsize LOP}}$ for $T \leq t+\tau$ on the RHS,
and then optimizes this at each time $s$ with respect to $X^{
\mbox{\scriptsize u}}$ (being the set of control parameters). We must check
that $\Delta_{\mbox{\scriptsize LOP}}$ is a solution to the HJB, \emph{and} that
the maximum on the RHS is realized when $X^{\mbox{\scriptsize u}}(t)$ is
precisely that prescribed by the LOP. Performing the substitution, we find
that the RHS is
\begin{equation}
\max_{X^{\mbox{\scriptsize u}}(t)} \left[ \zeta (t) |X^{\mbox{\scriptsize u}
}_{10}|^2 + |X^{\mbox{\scriptsize u}}_{20}|^2 + \eta(t) |X^{
\mbox{\scriptsize u}}_{21}|^2 \right] (8k \Lambda_2),
\end{equation}
where $\zeta = (\Lambda_1/\Lambda_2)e^{-\gamma(T-t)} \geq 1$ and $\eta(t)
\leq 1$. Note that $\gamma$ is already fixed by the LOP, and thus does not
take part in the optimization. We performed this maximization over $X^{
\mbox{\scriptsize u}}$ numerically, and verified that whenever $\zeta \geq 1
\geq \eta$, the maximum is obtained by the locally optimal $U$. Thus $|X^{
\mbox{\scriptsize u}}_{10}| = q$ and $|X^{\mbox{\scriptsize u}}_{20}| = |X^{
\mbox{\scriptsize u}}_{21}| = 0$. The RHS of the HJB equation is therefore $
\gamma \Lambda_1 e^{-\gamma(T-t)}$, and this is indeed equal to $\partial \Delta_{
\mbox{\scriptsize LOP}}/(\partial t)$, for $T \leq t + \tau$, being the LHS.
Since the derivatives of $\Delta_{\mbox{\scriptsize LOP}}$ that appear in the HJB
equation are all continuous for $T \leq t + \tau$ (the final requirement of the verification theorem), the LOP is the optimal
protocol for $T \leq t + \tau$.

To determine whether the LOP is optimal for $T \geq t+ \tau$, we note that the derivatives of $\Delta_{\mbox{\scriptsize LOP}}$ are \emph{not} continuous at $t = T- \tau$. As a result the classic verification theorem employed above no longer applies; we need a new verification theorem, developed in the last decade~\cite{Gozzi05},  that uses generalized solutions of second-order partial differential equations, referred to as \emph{viscosity solutions}~\cite{Crandall92}. Applying this ``viscosity" verification theorem to the LOP protocol for a qutrit shows that it remains optimal for $T > t+ \tau$. Since viscosity solutions will be unfamiliar to most readers, the details of this analysis will be presented elsewhere. We have also performed the analysis for the case when the eigenvalues of $X$ are not equally spaced, and in this case we find that the locally optimal protocol is {\em not} globally optimal. 

In summary, we have found the class of all locally optimal feedback
protocols in the regimes of good control and strong feedback, and obtained
explicit expressions for the best of these for $N=3$ and $
N=4$. We have also shown that the former is globally optimal for some, but not for all,  observables. The question of how to beat the LOP for a single qutrit when it is not optimal is an interesting one, and will be the subject of future work. 

{\em Acknowledgments:} We thank Michael Hsieh for suggesting functional analysis for  proving theorem 1, which was very useful, and one of the referees for simplifying our proof with the use of an upper bound.


\begin{thebibliography}
\expandafter\ifx\csname natexlab\endcsname\relax
\fi \expandafter\ifx\csname bibnamefont\endcsname\relax
\fi \expandafter\ifx\csname bibfnamefont\endcsname\relax
\fi \expandafter\ifx\csname citenamefont\endcsname\relax
\fi \expandafter\ifx\csname url\endcsname\relax
\fi \expandafter\ifx\csname urlprefix\endcsname\relax
\fi \providecommand{\bibinfo}[2]{#2} \providecommand{\eprint}[2][]{\url{#2}}

\vspace{-2ex}

\bibitem{Majer07} \bibinfo{author}{
\bibfnamefont{J.}~\bibnamefont{Majer}},
\emph{et~al.}, \bibinfo{journal}{Nature}  \textbf{\bibinfo{volume}{449}}, 
\bibinfo{pages}{443} (\bibinfo{year}{2007}).

\bibitem{Steffen06} 
\bibinfo{author}{\bibfnamefont{M.}~\bibnamefont{Steffen}},
\emph{et~al.}, \bibinfo{journal}{Science}  \textbf{\bibinfo{volume}{313}}, 
\bibinfo{pages}{1423} (\bibinfo{year}{2006}).

\bibitem{Niskanen07} 
\bibinfo{author}{\bibfnamefont{A.~O.}
\bibnamefont{Niskanen}},
\emph{et~al.}, \bibinfo{journal}{Science} \textbf{\bibinfo{volume}{316}},  
\bibinfo{pages}{723} (\bibinfo{year}{2007}).

\bibitem{Plantenberg07} 
\bibinfo{author}{\bibfnamefont{J.~H.}
\bibnamefont{Plantenberg}},  \bibinfo{author}{\bibfnamefont{P.~C.}
\bibnamefont{de~Groot}},  \bibinfo{author}{\bibfnamefont{C.~J. P.~M.}
\bibnamefont{Harmans}},  and \bibinfo{author}{\bibfnamefont{J.~E.}
\bibnamefont{Mooij}},  \bibinfo{journal}{Nature} \textbf{\bibinfo{volume}{447}},  
\bibinfo{pages}{836} (\bibinfo{year}{2007}).

\bibitem[Boixo et~al.(2007)Boixo, Flammia, Caves, and Geremia]{Boixo07} 
\bibinfo{author}{\bibfnamefont{S.}~\bibnamefont{Boixo}},  
\bibinfo{author}{\bibfnamefont{S.~T.} \bibnamefont{Flammia}},  
\bibinfo{author}{\bibfnamefont{C.~M.} \bibnamefont{Caves}}, and  
\bibinfo{author}{\bibfnamefont{JM}~\bibnamefont{Geremia}},  
\bibinfo{journal}{Phys. Rev. Lett.} \textbf{\bibinfo{volume}{98}},  
\bibinfo{pages}{090401} (\bibinfo{year}{2007}).

\bibitem{Cook07}
R. L. Cook, P. J. Martin and J. M. Geremia, Nature {\bf 446}, 774 (2007). 

\bibitem{Chen07}
\bibinfo{author}{\bibfnamefont{S.}~\bibnamefont{Chen}},
\emph{et~al.}, \bibinfo{journal}{Phys. Rev. Lett.} \textbf{
\bibinfo{volume}{99}},  \bibinfo{pages}{180505} (\bibinfo{year}{2007}).

\bibitem[Jacobs(2003)]{rapidP} \bibinfo{author}{\bibfnamefont{K.}~
\bibnamefont{Jacobs}},  \bibinfo{journal}{Phys. Rev. A} \textbf{
\bibinfo{volume}{67}},  \bibinfo{pages}{030301(R)} (\bibinfo{year}{2003}).

\bibitem[Smith et~al.(2002)Smith, Reiner, Orozco, Kuhr, and Wiseman]
{Smith02} \bibinfo{author}{\bibfnamefont{W.~P.} \bibnamefont{Smith}},  
\emph{et~al.}, \bibinfo{journal}{{Phys.\ Rev.\ Lett.}} \textbf{\bibinfo{volume}{89}},  
\bibinfo{pages}{133601} (\bibinfo{year}{2002}).

\bibitem[Steck et~al.(2004)Steck, Jacobs, Mabuchi, Bhattacharya, and Habib]
{Steck04} \bibinfo{author}{\bibfnamefont{D.}~\bibnamefont{Steck}},  
\emph{et~al.}, \bibinfo{journal}{Phys. Rev. Lett.} \textbf{\bibinfo{volume}{92}},  
\bibinfo{pages}{223004} (\bibinfo{year}{2004}).

\bibitem{Yanagisawa06} \bibinfo{author}{\bibfnamefont{M.}~
\bibnamefont{Yanagisawa}},  \bibinfo{journal}{Phys. Rev. Lett.} \textbf{
\bibinfo{volume}{97}},  \bibinfo{pages}{190201} (\bibinfo{year}{2006}).

\bibitem{Yanagisawa98}
 M. Yanagisawa, and H.  Kimura, in {\em Learning, Control and Hybrid Systems}, 
Lecture Notes in Control and Information Sciences {\bf 241}, 249, (Springer, New York, 1998); \bibinfo{author}{\bibfnamefont{A.~C.}
\bibnamefont{Doherty}} and  \bibinfo{author}{\bibfnamefont{K.}~
\bibnamefont{Jacobs}},  \bibinfo{journal}{Phys.\ Rev.\ A} \textbf{
\bibinfo{volume}{60}},  \bibinfo{pages}{2700} (\bibinfo{year}{1999});
\bibinfo{author}{\bibfnamefont{V.~P.}
\bibnamefont{Belavkin}},  \bibinfo{journal}{Rep. Math. Phys.} \textbf{
\bibinfo{volume}{43}},  \bibinfo{pages}{405} (\bibinfo{year}{1999}); 
C. D'Helon and M. R. James, Phys. Rev. A 73, 053803 (2006).  

\bibitem{Bhattacharya00} 
\bibinfo{author}{\bibfnamefont{S.}~\bibnamefont{Habib}},  
\bibinfo{author}{\bibfnamefont{K.}~\bibnamefont{Jacobs}}, and  
\bibinfo{author}{\bibfnamefont{K.}~\bibnamefont{Shizume}},  
\bibinfo{journal}{Phys. Rev. Lett.} \textbf{\bibinfo{volume}{96}},  
\bibinfo{pages}{010403} (\bibinfo{year}{2006}).

\bibitem[Whittle(1996)]{Whittle} \bibinfo{author}{\bibfnamefont{P.}~
\bibnamefont{Whittle}},  \emph{\bibinfo{title}{Optimal Control}} 
(\bibinfo{publisher}{Wiley, Chichester}, \bibinfo{year}{1996}).

\bibitem[Jacobs and Lund(2007)]{Jacobs07c} \bibinfo{author}{
\bibfnamefont{K.}~\bibnamefont{Jacobs}} and  \bibinfo{author}{
\bibfnamefont{A.~P.} \bibnamefont{Lund}},  \bibinfo{journal}{Phys. Rev.
Lett.} \textbf{\bibinfo{volume}{99}},  \bibinfo{pages}{020501} (
\bibinfo{year}{2007}).

\bibitem{Vion02}
\bibinfo{author}{\bibfnamefont{D.}~\bibnamefont{Vion}, {\em et al.}},
  \bibinfo{journal}{Science} \textbf{\bibinfo{volume}{296}},
  \bibinfo{pages}{886} (\bibinfo{year}{2002}).

\bibitem{Jacobs07b} 
   K. Jacobs and P. Lougovski and  M. P. Blencowe, Phys. Rev. Lett. {\bf 98}, 147201 (2007). 

\bibitem{rap2} 
H. M. Wiseman and J. F. Ralph, New. J. Phys {\bf 8}, 90 (2006); 
J. Combes and K. Jacobs, Phys. Rev. Lett. {\bf 96}, 010504 (2006);
J. Combes, H. M. Wiseman, and K. Jacobs, Phys. Rev. Lett. {\bf 100}, 
160503 (2008).

\bibitem{Seo} 
   S. Izumino, H. Mori and Y. Seo, J. Inequal. \& Appl. {\bf 2}, 235 (1998). 

\bibitem[Bransden and Joachain(1989)]{BJBook} \bibinfo{author}{
\bibfnamefont{B.~H.} \bibnamefont{Bransden}} and  \bibinfo{author}{
\bibfnamefont{C.~J.} \bibnamefont{Joachain}},  \emph{
\bibinfo{title}{Introduction to Quantum Mechanics}} 
(\bibinfo{publisher}{Wiley, Chichester}, \bibinfo{year}{1989}).

\bibitem{Wiseman08}
\bibinfo{author}{\bibfnamefont{H.~M.} \bibnamefont{Wiseman}} \bibnamefont{and}
 \bibinfo{author}{\bibfnamefont{L.}~\bibnamefont{Bouten}},
  \bibinfo{journal}{Quant. Inform. Proc.} \textbf{\bibinfo{volume}{7}},
  \bibinfo{pages}{71} (\bibinfo{year}{2008}).
  
\bibitem[Zhou et~al.(1997)Zhou, Yong, and Li]{Zhou97} \bibinfo{author}{
\bibfnamefont{X.~Y.} \bibnamefont{Zhou}},  \bibinfo{author}{
\bibfnamefont{J.}~\bibnamefont{Yong}}, and  \bibinfo{author}{
\bibfnamefont{X.}~\bibnamefont{Li}},
\bibinfo{journal}{SIAM
  J. Cont. Optim.} \textbf{\bibinfo{volume}{35}}, \bibinfo{pages}{243}  (
\bibinfo{year}{1997}).

\bibitem[Gozzi et~al.(1997)Gozzi, Swiech, and Zhou]{Gozzi05} 
\bibinfo{author}{\bibfnamefont{F.}~\bibnamefont{Gozzi}},  
\bibinfo{author}{\bibfnamefont{A.}~\bibnamefont{Swiech}}, and  
\bibinfo{author}{\bibfnamefont{X.~Y.} \bibnamefont{Zhou}},  
\bibinfo{journal}{SIAM J. Cont. Optim.} \textbf{\bibinfo{volume}{43}},  
\bibinfo{pages}{2009} (\bibinfo{year}{2005}).

\bibitem[Crandall et~al.(1992)Crandall, Ishii, and Lions]{Crandall92} 
\bibinfo{author}{\bibfnamefont{M.~G.} \bibnamefont{Crandall}},  
\bibinfo{author}{\bibfnamefont{H.}~\bibnamefont{Ishii}}, and  
\bibinfo{author}{\bibfnamefont{P.-L.} \bibnamefont{Lions}},  
\bibinfo{journal}{B. Am. Math. Soc.} \textbf{\bibinfo{volume}{27}},  
\bibinfo{pages}{1} (\bibinfo{year}{1992}).
\end{thebibliography}
\end{document}